\pdfoutput=1
\documentclass[pra,twocolumn,longbibliography,nofootinbib]{revtex4-1} 
\usepackage{amsmath,amssymb,amsthm,amscd,amsxtra,amsfonts,mathrsfs,graphicx,enumerate,bm,slashed,array,mathtools,enumitem}
\usepackage{xcolor}
\usepackage[all,color]{xy}
\setcitestyle{numbers}
\input xy
\xyoption{all}
\newtheorem{Theorem}{Theorem}[section]
\newtheorem{Lemma}[Theorem]{Lemma}
\newtheorem{Proposition}[Theorem]{Proposition}

\newtheorem{Example}[Theorem]{Example}
\newtheorem{Remark}[Theorem]{Remark}
\newtheorem{Question}[Theorem]{Question}
\newtheorem{Definition}[Theorem]{Definition}

\newtheorem*{Theorem A}{Theorem A}
\begin{document}
 \author{Charlie Beil}
 %\thanks{The author was supported by , which he gratefully acknowledges.}
 %\address{Heilbronn Institute for Mathematical Research, School of Mathematics, Howard House, The University of Bristol, Bristol, BS8 1SN, United Kingdom.}
 \affiliation{Heilbronn Institute for Mathematical Research, Howard House, The University of Bristol, Bristol, BS8 1SN, United Kingdom.}
 \email{charlie.beil@bristol.ac.uk}
 \title[On the spacetime geometry of quantum nonlocality]{On the spacetime geometry of\\ quantum nonlocality}
%\subjclass[2010]{13C15, 14A20}
 %\date{}

\begin{abstract}
We present a new geometry of spacetime where events may be positive dimensional.
This geometry is obtained by applying the identity of indiscernibles, which is a fundamental principle of quantum statistics, to time.
Quantum nonlocality arises as a natural consequence of this geometry. %, and avoids retrocausality. 
We also examine the ontology of the wavefunction in this framework.
In particular, we show how entanglement swapping in spacetime invalidates the preparation assumption of the PBR theorem.
\end{abstract}

\maketitle

\section{Introduction}

An elementary discrepancy between quantum theory and relativity is that quantum theory is inherently nonlocal, whereas spacetime has the structure of a manifold, and is thus local by construction. %defined locally. 
The discrepancy is resolved on the level of information, since the intrinsic randomness in the measurement of a quantum state prevents instantaneous signaling (by the no-communication theorem \cite[II.E]{PT}).
This resolution is satisfactory if information is considered to be fundamental \cite[III.C]{PT}.
However, if one considers geometry to be fundamental, then the discrepancy remains.

Here we pursue a possible resolution from the perspective that geometry is fundamental, with the aim that it may shed light on the nature of quantum gravity.\footnote{To say that geometry is fundamental, we do not mean that the block universe has an objective physical reality; see Remark \ref{block}.} 
Just as simultaneity has no universal meaning in special relativity, we propose that a `moment of time' has no universal meaning, and different observers will in general disagree about the `duration' of a single moment of time. 
In particular, even clocks in the same inertial frame may disagree.

The paper is organized as follows.
We first propose a new operational definition of time using the identity of indiscernibles: we postulate that time passes if and only if a system undergoes a transformation which is not local and invertible.
We then show that this postulate is compatible with the thermodynamic arrow of time in a generic example.
Furthermore, the postulate results in a spacetime with positive dimensional events, thus giving rise to Bell nonlocality without requiring retrocausality.
Finally, we examine the ontology of the wavefunction in this framework.
In particular, we show that if spacetime events are topologically closed, then the wavefunction is epistemic.
Moreover, we find that the preparation assumption of the PBR theorem does not hold using the worldlines of 4-photon entanglement swapping.
 
We remark that our work pertains to Minkowski spacetime; we speculate that a mathematical formulation of nonlocal curved spacetime may be a hybrid of differential geometry and nonnoetherian algebraic geometry \cite{B}, the latter of which was introduced in \cite{B1} to study a class of non-superconformal quiver gauge theories in string theory.

Throughout, $\left| 0 \right\rangle$ and $\left| 1 \right\rangle$ form an orthonormal basis for the qubit Hilbert space, and $\left| \pm \right\rangle := \frac{1}{\sqrt{2}} \left( \left| 0 \right\rangle \pm \left| 1 \right\rangle \right)$.

\section{Time and the identity of indiscernibles} \label{first section}

To define a notion of spacetime that encompasses quantum nonlocality, we must first define what a clock is.
In this section we introduce an operational definition of time using the identity of indiscernibles.

The identity of indiscernibles is a principle, due to Leibniz, which states that two distinct objects cannot be identical.
More precisely, for each $x$ and $y$, if $x$ and $y$ have all the same properties, $\forall P (Px \Leftrightarrow Py)$, then $x = y$.
This principle is fundamental in quantum statistics (both Bose-Einstein and Fermi-Dirac). 
Applying the identity of indiscernibles to time, we obtain:
\begin{center}
\textit{If a system does not change, then the system does not experience time.}
\end{center}
We call this assumption, together with its converse, the `weak postulate of time':
\begin{center}
\textit{A system experiences time if and only if the system undergoes a non-trivial transformation.}
\end{center}

We assume that the experience of time is transitive: If a system $\mathcal{H}_1$ experiences time under a transformation $T_1: \mathcal{H}_1 \to \mathcal{H}_2$, and $\mathcal{H}_2$ experiences time under a transformation $T_2: \mathcal{H}_2 \to \mathcal{H}_3$, then $\mathcal{H}_1$ experiences time under the composition $T_2T_1: \mathcal{H}_1 \to \mathcal{H}_3$.

\begin{Lemma} \label{non-trivial}
If the weak postulate of time holds, then non-trivial invertible transformations cannot exist.
\end{Lemma}

\begin{proof}
Suppose to the contrary that a system $\mathcal{H}_1$ undergoes a sequence of non-trivial transformations 
$$\mathcal{H}_1 \stackrel{T_1}{\longrightarrow} \mathcal{H}_2 \stackrel{T_2}{\longrightarrow} \mathcal{H}_1$$
such that the composition $T_2T_1$ is the identity on $\mathcal{H}_1$.
Then the system will experience time under $T_1$ and $T_2$, but not under their composition, in contradiction to transitivity.
\end{proof}

Quantum systems evolve by unitary transformations.
Thus, in view of Lemma \ref{non-trivial}, the weak postulate must be modified.

\begin{Definition} \rm{
Let $T: \mathcal{H} \to \mathcal{H}$ be a linear transformation.
We say $T$ is \textit{local} with respect to a tensor product decomposition $\mathcal{H} \cong \bigotimes_{i \in I} \mathcal{H}_i$ if $T$ decomposes
$$T = \bigotimes_{i\in I} T_i, \ \ \text{ with } \ T_i: \mathcal{H}_i \to \mathcal{H}_i.$$
We say $T$ is \textit{local invertible (LI)} if $T$ is local and invertible; otherwise we say $T$ is \textit{non-LI}.
}\end{Definition}

We propose the following `strong postulate of time':
\begin{center}
\textit{A system experiences time if and only if the system undergoes a non-LI transformation.} 
\end{center}

From this postulate we obtain an abstract definition of a clock.

\begin{Definition} \label{clock} \rm{
A clock is a system that undergoes a sequence of non-LI transformations 
$$T_i: \mathcal{H}_i \to \mathcal{H}_{i+1}, \ \ \ i \in \mathbb{Z}.$$ 
}\end{Definition}

We note that the strong postulate of time (without the local assumption) is morally similar to a form of spontaneous symmetry breaking.
Indeed, if a transformation $T: \mathcal{H} \to \mathcal{H}$ of Hilbert spaces is not invertible, then the rank of $T$ is less than the dimension of $\mathcal{H}$,
$$\operatorname{dim}T(\mathcal{H}) < \operatorname{dim}\mathcal{H}.$$
Thus the rank of the general linear group of $T(\mathcal{H})$ is less than the rank of the general linear group of $\mathcal{H}$,
$$\operatorname{rank}\left( \text{GL}\left( T(\mathcal{H}) \right) \right) < \operatorname{rank}\left(\operatorname{GL}\left( \mathcal{H} \right) \right).$$
In particular, $T(\mathcal{H})$ has less symmetry than $\mathcal{H}$.
We may therefore view $T$ as spontaneously breaking the symmetries of $\mathcal{H}$, and the property that emerges under $T$ is a new moment of time.

\section{The Hilbert space of a single photon in an otherwise empty universe} \label{Hilbert space section}

By the strong postulate of time, a system $\mathcal{H}$ does not experience time if $\mathcal{H}$ undergoes an LI transformation $T: \mathcal{H} \to \mathcal{H}$.
One may ask, if $T$ is non-trivial, then how can $\mathcal{H}$ not experience time under $T$, since such a transformation is still an honest change of $\mathcal{H}$? 
In this section we address this question. 

The Hilbert space of a state is generated, as a vector space, by its physical environment.\footnote{Here we do not use density matrices, as density matrices describe either an ensemble of states, or repeated runs of identically prepared states, and we are interested in what happens in a single quantum event.}
Specifically, the Hilbert space $\mathcal{H}_E$ of a quantum state $\psi$ in an environment $E$ is generated by the quantum states $\phi_1, \ldots, \phi_n$ for which the probability $c_i$ of transition $\psi \mapsto \phi_i$ is nonzero in $E$.
In particular,
\begin{equation} \label{ci}
\mathcal{H}_E = \operatorname{span}\left\{ \left| \phi_1 \right\rangle, \ldots, \left| \phi_n \right\rangle \right\} \ \ \text{ and } \ \ \left\langle \phi_i | \psi \right\rangle^2 = c_i.
\end{equation}
We call the set of vectors $\left\{ \left| \phi_i \right\rangle \right\}$ the \textit{physical spanning set} for $\psi$.

`Non-separable' is usually taken to be synonymous with `entangled', and so we introduce the following terminology for clarity. 

\begin{Definition} \rm{
We say a pure state $\left| \psi \right\rangle \in \mathcal{H}_E$ is \textit{indivisible} if it is not separable with respect to its physical spanning set.
}\end{Definition}

An example of a non-entangled indivisible state is a single photon or qubit. 
We propose the following:
\begin{center}
\textit{An indivisible state exchanges information with its environment if and only if the environment acts on the state by a non-LI transformation.}
\end{center}

By `measurement' we mean an exchange of information between a state and its environment. 
This may be realized by either wavefunction collapse or non-unitary entanglement.
Wavefunction collapse itself is not necessarily a non-invertible transformation, as the following example demonstrates.

\begin{Example} \label{first example} \rm{
Consider a photon with initial polarization $\left| + \right\rangle$ encountering a polarizer with orientation $\left| 0 \right\rangle$. 
The photon may either be absorbed by the polarizer, in which case its wavefunction collapses onto $\left| 1 \right\rangle$; or the photon may pass through, in which case its wavefunction collapses onto $\left| 0 \right\rangle$.
In the latter case, the original polarization of the photon may be restored by passing the photon through a Hadarmard gate $H = \left| + \right\rangle \left\langle 0 \right| + \left| - \right\rangle \left\langle 1 \right|$.
In particular, the collapse $\left| + \right\rangle \mapsto \left| 0 \right\rangle$ is an invertible transformation.
}\end{Example}

The Hilbert space $\mathcal{H}_E$ of $\psi$ depends on the choice of environment $E$.
It is often the case (though perhaps not always) that the dimension of $\mathcal{H}_E$ increases as the environment $E$ of $\psi$ is enlarged,
\begin{equation} \label{dim HE}
\operatorname{dim}\mathcal{H}_E \leq \operatorname{dim}\mathcal{H}_{E'} \ \ \text{ whenever } \ \ E \subseteq E'.
\end{equation}

\begin{Example} \rm{
Consider a photon encountering a sequence of $n$ polarizers $P_1, \ldots, P_n$, which we may take to be electrons oscillating at the respective angles $\theta_1, \ldots, \theta_n$.
At each polarizer $P_i$, the photon is either absorbed, in which case its wavefunction collapses onto $\left| \theta_i \right\rangle \in \mathcal{H}_i \cong \mathbb{C}^2$, 
or passes through, in which cases its wavefunction collapses onto $\left| \frac{\pi}{2} - \theta_i \right\rangle \in \mathcal{H}_i$.
For each $1 \leq m < n$, denote by $E(m)$ the environment consisting of the first $m$ polarizers $P_1, \ldots, P_m$.
The Hilbert space $\mathcal{H}_{E(m)}$ of $\psi$ is then
$$\mathcal{H}_{E(m)} = \operatorname{span}\left\{ \left| 1,\theta_1 \right\rangle, \ldots, \left| m,\theta_m \right\rangle, \left| m, \pi/2- \theta_m \right\rangle \right\}.$$
Furthermore, $\left\langle i,\theta_i | j,\theta_j \right\rangle = 0$ for $i \not = j$ since the probability of transition $\left| i,\theta_i \right\rangle \mapsto \left| j,\theta_j \right\rangle$ is zero.
Therefore for a generic choice of $\theta_i$,
$$\mathcal{H}_{E(m)} \cong \mathbb{C}^{m+1}.$$
In particular, (\ref{dim HE}) holds.
}\end{Example}

We may extrapolate the environmental dependence of $\mathcal{H}$ down to $\psi$ itself, where the environment of $\psi$ consists of only $\psi$.
Indeed, by the definition of $\mathcal{H}_E$ given in (\ref{ci}) and the fact the probability of transition $\psi \mapsto \psi$ is 1, we propose that \textit{the Hilbert space of an indivisible state $\psi$ with respect to the environment $E = \left\{ \psi \right\}$ is one-dimensional,}
\begin{equation} \label{H psi}
\mathcal{H}_{\psi} = \operatorname{span}\left\{ \left| \psi \right\rangle \right\} \cong \mathbb{C}.
\end{equation}
This may be interpreted to mean that all the properties of an indivisible state are \textit{relational}; without reference to the exterior universe, a quantum state has no intrinsic properties.

\begin{Remark} \rm{
Again consider a photon undergoing an invertible transformation as in Example \ref{first example}.
The idea behind (\ref{H psi}) is that a photon (as well as any indivisible state) \textit{does not detect its environment when there is no exchange of information}.  
Thus from the photons perspective, it lives completely alone in an otherwise empty universe.
Changing the photons polarization by an invertible transformation is then similar to rotating the photon together with its entire universe, and thus goes unnoticed by the photon.
An observer outside the universe (or an experimenter in an optics lab) would infer this change, but the photon would not. 
}\end{Remark}

In the following proposition we resolve the question raised at the beginning of this section.

\begin{Proposition} \label{one-dim}
The weak postulate of time is equivalent to the strong postulate of time for indivisible states.
\end{Proposition}

\begin{proof} 
Since the Hilbert space $\mathcal{H}_{\psi} \cong \mathbb{C}$ of an indivisible state $\psi$ is one-dimensional, there are only two morphisms $\mathcal{H}_{\psi} \to \mathcal{H}_{\psi}$ up to a global phase factor: the identity map and the zero map.\footnote{In the language of quivers, we are considering representation isoclasses of the quiver with two vertices and one arrow between them, of dimension vector $(1,1)$, rather than one-dimensional representation isoclasses of the quiver with one vertex and one loop.} 

Consider an invertible transformation $T: \mathcal{H}_E \to \mathcal{H}_E$ of a composite system of indivisible states.
Then $T$ extends to the commutative diagram,
\begin{equation} \label{diagram}
\xymatrix{ \mathcal{H}_E \ar[r]^T \ar[d]_{\pi_{\left| \psi \right\rangle}} & \mathcal{H}_E \ar[d]^{\pi_{T\left| \psi \right\rangle}}
\\ \mathcal{H}_{\psi} \ar[r]^{\cdot 1} & \mathcal{H}_{\psi}}
\end{equation}
where the maps $\pi_{\left| \psi \right\rangle}$ and $\pi_{T\left| \psi \right\rangle}$ are the projections from $\mathcal{H}_E$ to the one-dimensional subspaces spanned by $\left| \psi \right\rangle$ and $T \left| \psi \right\rangle$ respectively.
\textit{Therefore $T$ appears to be trivial to the indivisible state $\psi$.}

On the other hand, suppose $\left| \psi \right\rangle$ is in the kernel of $T$ (for example, if the photon is absorbed by the polarizer in Example \ref{first example}).
Then both the composite system $E$ and the subsystem $\psi$ experience time by the weak postulate of time.
Indeed, $T$ extends to a commutative diagram as in (\ref{diagram}),
$$\xymatrix{ \mathcal{H}_E \ar[r]^{T} \ar[d]_{\pi_{\left| \psi \right\rangle}} & \mathcal{H}_E \ar[d]^{\pi_{T\left| \psi \right\rangle} = 0}
\\ \mathcal{H}_{\psi} \ar[r]^{\cdot 0} & 0}$$
\end{proof}

\section{Thermodynamic entropy and the strong postulate of time} \label{thermodynamic}

In this section we consider a simple thermodynamic system, and ask whether the strong postulate of time implies the experience of time as entropy increases. %that the system experiences time as its entropy increases.

Consider a finite set of distinct particles $\mathcal{P}$ moving in Minkowski space $X$, and a finite set of points $\mathcal{Q} \subset X$, called sites, in a fixed inertial frame.
Denote by $\left| \mathcal{P} \right|$ the number of particles, and by $\left| \mathcal{Q} \right|$ the number of sites.
Set 
$$\mathcal{P}' := \mathcal{P} \cup \left\{ \emptyset \right\}.$$

Let $t \in \mathbb{R}$ be the time parameter in the inertial frame of $\mathcal{Q}$.
Suppose at times $t = 0$ and $t = 1$ the particles lie on the sites of $\mathcal{Q}$.
Then at these times, the location of each particle may be specified by one of the maps
$$\sigma_t: \mathcal{P} \longrightarrow \mathcal{Q} \ \ \ \text{ and } \ \ \ \tau_t: \mathcal{Q} \longrightarrow \mathcal{P}',$$
where $\sigma_t$ specifies the site of each particle, and $\tau_t$ specifies the particle at each site. 
Note that for each particle $p \in \mathcal{P}$,
$$\tau_t \sigma_t(p) = p.$$
Furthermore, if no particle is located at site $q \in \mathcal{Q}$ at time $t$, then $\tau_t(q) = \emptyset$.

Denote by $\mathbb{C}\mathcal{P}$ and $\mathbb{C}\mathcal{Q}$ the $\mathbb{C}$-vector spaces freely generated by $\mathcal{P}$ and $\mathcal{Q}$, respectively.
The transformation of the particle system from $t = 0$ to $t = 1$ may be represented by two $\mathbb{C}$-linear maps.
In the first case, the particles are fixed and the sites are transformed:
$$\begin{array}{crcl}
S: & \mathbb{C}\mathcal{Q}^{\otimes \left| \mathcal{P} \right|} & \longrightarrow & \mathbb{C}\mathcal{Q}^{\otimes \left| \mathcal{P} \right|} \\
& \bigotimes_{p \in \mathcal{P}} \left| \sigma_0(p) \right\rangle & \mapsto & \bigotimes_{p \in \mathcal{P}} \left| \sigma_1(p) \right\rangle
\end{array}$$
In the second case, the sites are fixed and the particles are transformed:
$$\begin{array}{crcl}
T: & \mathbb{C}\mathcal{P}'^{\otimes \left| \mathcal{Q} \right|} & \longrightarrow & \mathbb{C}\mathcal{P}'^{\otimes \left| \mathcal{Q} \right|} \\
& \bigotimes_{q \in \mathcal{Q}} \left| \tau_0(q) \right\rangle & \mapsto & \bigotimes_{q \in \mathcal{Q}} \left| \tau_1(q) \right\rangle
\end{array}$$

Suppose that at least one particle changes sites.
Then the \textit{system of particles} undergoes a local transformation under $S$, but not under $T$.
In particular, with respect to the particle system, $T$ is a non-LI transformation.
In contrast, if distinct particles are on distinct sites, then $S$ is an LI transformation. % if distinct particles are on distinct sites, whereas $T$ is non-LI.

Furthermore, in the $S$ representation a clock may be assigned to each particle, and in the $T$ representation a clock may be assigned to each site.
By relativity, clocks assigned to different particles will generically disagree since the particles are in relative motion.
Thus there is no global notion of time in the $S$ representation. 
In contrast, clocks assigned to different sites will agree since the sites belong to a single inertial frame. 
Therefore there is a global notion of time in the $T$ representation. 

Consequently, to determine if time has passed for the system of particles, the $T$ representation must be used.
Since $T$ is a non-LI transformation, the system of particles experiences time by the strong postulate of time.
The resulting implication 
\begin{center}
\textit{entropy increases $\Longrightarrow$ system experiences time} 
\end{center}
is thus a consequence of the strong postulate of time together with special relativity.\footnote{Of course, the converse implication is the content of the second law of thermodynamics.} %, as special relativity invalidates the $S$ representation as useful for measuring the passage of time.

\section{Positive dimensional spacetime events from quantum clocks} \label{positive dimensional section}

To distinguish between classical (relativistic) spacetime and the spacetime we will introduce, we we will always use the adjective `inferred' when referring to classical spacetime; for example, \textit{inferred spacetime}, \textit{inferred event}, \textit{inferred proper time}, \textit{inferred metric}, etc.
Recall that an inferred event is a zero-dimensional point in inferred spacetime.

\begin{Definition} \rm{
We define spacetime as follows.
\begin{enumerate} [label=(\alph*)]
\item The \textit{support} of an indivisible state $\psi$ is the locus of points in inferred spacetime where it is possible in principle to measure $\psi$. 
 \item A \textit{spacetime event} is the support of an indivisible state.
\end{enumerate}
}\end{Definition}

\begin{Lemma}
The definition of a spacetime event is well-defined.
\end{Lemma}

\begin{proof}
An indivisible state cannot be transformed into a separable state by an LI transformation since a local unitary transformation $U_1 \otimes U_2 \in \operatorname{GL}\left( \mathcal{H}_1 \otimes \mathcal{H}_2 \right)$ preserves the separability of a state in $\mathcal{H}_1 \otimes \mathcal{H}_2$.
(Of course, any entangled state can be rotated to a separable state by a \textit{global} unitary transformation.)
\end{proof}

\begin{Remark} \rm{
The time and space parameters in the Schr\"odinger, Klein-Gordon, and Dirac equations are all parameters in inferred spacetime.
Consequently, these propagation equations are actually part classical, part quantum, and thus do not give a fully quantum description.
}\end{Remark}

\begin{Remark} \label{block} \rm{
Our definition of spacetime does not imply the reality of a timeless block universe (and thus the absence of free will), just as it is not implied by inferred spacetime. 
Indeed, in both theories spacetime may be regarded as epistemic, that is, as simply a mathematical construction. 
}\end{Remark}

\begin{Question} \rm{
Can spacetime events intersect nontrivially in inferred spacetime?
This question will be considered (though not answered) in Proposition \ref{overlap} below.
}\end{Question}

Time may emerge for a composite system whose constituent parts do not experience time.
Indeed, consider a system of indivisible states, such as the system of particles considered in Section \ref{thermodynamic}.
If the system undergoes a non-LI transformation, then the system as a whole experiences time by the strong postulate of time, whereas the individual states do not.
We conclude that there are two kinds of time that arise from the strong postulate of time:
\begin{enumerate}
 \item \textit{Fundamental time}, which results from distinct spacetime events along an inferred worldline.
 \item \textit{Emergent time}, which results from a non-LI transformation of a composite system.
\end{enumerate}

\begin{Example} \rm{
Consider again Example \ref{first example}, where a photon passes through a polarizer and thus undergoes an LI transformation.
Suppose the polarizer is an electron oscillating orthogonal to the direction of polarization.
The composite system consisting of both the photon and the electron experiences time by the argument given in Section \ref{thermodynamic}.
However, the photon and electron individually do not experience time. 
The time that the composite system experiences is therefore emergent. 
}\end{Example}

\section{Spacetime support of entangled states} \label{entanglement section}

\begin{figure}
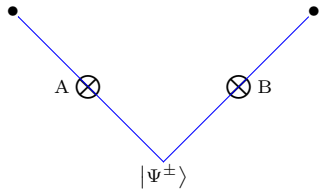

$$\xy
(0,-12)*{\text{\scriptsize{$\left|\Psi^{\pm} \right\rangle$}}}="";
(0,-10)*{}="1";(-10,0)*{\bigotimes}="2";(10,0)*{\bigotimes}="3";
(-13.5,0)*{\text{\scriptsize{A}}}="";(13.5,0)*{\text{\scriptsize{B}}}="";
(-20,10)*{\bullet}="4";(20,10)*{\bullet}="5";
{\ar@{-}@[blue]"1";"4"};{\ar@{-}@[blue]"1";"5"};
\endxy$$
%$$\xy
%(0,-12)*{\text{\scriptsize{$\left|\Psi^{\pm} \right\rangle$}}}="";
%(0,-10)*{}="1";(-10,0)*{\bigotimes}="2";(10,0)*{\bigotimes}="3";
%(-13.5,0)*{\text{\scriptsize{A}}}="";(13.5,0)*{\text{\scriptsize{B}}}="";
%(-20,10)*{\bullet}="4";(20,10)*{\bullet}="5";
%{\ar@{-}"1";"4"};{\ar@{-}"1";"5"};
%\endxy$$
\caption{The spacetime diagram for the Bell states, in the center-of-mass rest frame.
The horizontal and vertical axes are inferred space and time coordinates, respectively.}
\label{Bell}
\end{figure}

In this section we consider spacetime events of maximally entangled qubits with well-defined inferred positions (e.g., particles with well-defined inferred worldlines).
Recall that maximal entanglement is monogamous, that is, if two qubits are maximally entangled, then they cannot be entangled with a third qubit \cite{CKW}.
We leave the more difficult case of non-maximal entanglement for future work.

\subsection{Bell states}

The four maximally entangled Bell states are
$$\left| \Phi^{\pm} \right\rangle := \frac{1}{\sqrt{2}} \left( \left| 0 0 \right\rangle \pm \left| 1 1 \right\rangle \right), \ \ \ \ \left| \Psi^{\pm} \right\rangle := \frac{1}{\sqrt{2}} \left( \left| 0 1 \right\rangle \pm \left| 1 0 \right\rangle \right).$$
The spacetime diagram of a Bell state is given in Figure \ref{Bell}, where the two particles are flying away from each other. 
The state is supported on the union of the inferred wordlines of the two particles, with boundaries at the inferred events where the particles interact non-invertibly with their respective environments. 
This V-shaped support is then a single event in spacetime by the strong postulate of time.

Furthermore, if one of the two particles passes through a polarizer, then the transformation on the state is LI.
Therefore the original spacetime event continues to be traced out by the two particles, thereby allowing Bell's inequality to be violated.

\begin{figure}
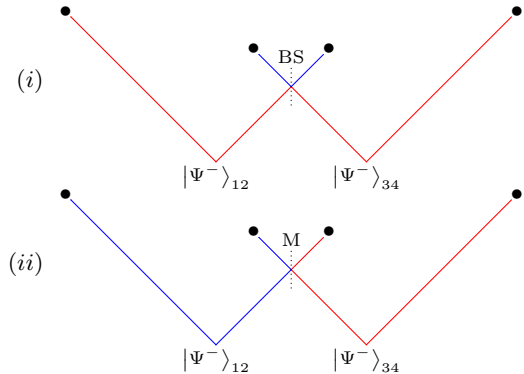

$$\begin{array}{rc}
(i) & \xy
(-10,-12)*{\text{\scriptsize{$\left|\Psi^- \right\rangle_{12}$}}}="";
(10,-12)*{\text{\scriptsize{$\left|\Psi^- \right\rangle_{34}$}}}="";
(0,4)*{\text{\scriptsize{BS}}}="";
(-10,-10)*{}="1";
(10,-10)*{}="2";
(-20,0)*{}="3";
(0,0)*{}="4";
(20,0)*{}="5";
(-30,10)*{\bullet}="6";
(-5,5)*{\bullet}="7";(5,5)*{\bullet}="8";
(30,10)*{\bullet}="9";
(0,-2.5)*{}="11";(0,2.5)*{}="10";
{\ar@{-}@[red]"1";"6"};{\ar@{-}@[red]"1";"4"};{\ar@{-}@[red]"2";"4"};{\ar@{-}@[red]"2";"9"};
{\ar@{-}@[blue]"4";"7"};{\ar@{-}@[blue]"4";"8"};{\ar@{..}"11";"10"};
\endxy
\\
(ii) & 
\xy
(-10,-12)*{\text{\scriptsize{$\left|\Psi^- \right\rangle_{12}$}}}="";
(10,-12)*{\text{\scriptsize{$\left|\Psi^- \right\rangle_{34}$}}}="";
(0,4)*{\text{\scriptsize{M}}}="";
(-10,-10)*{}="1";
(10,-10)*{}="2";
(-20,0)*{}="3";
(0,0)*{}="4";
(20,0)*{}="5";
(-30,10)*{\bullet}="6";
(-5,5)*{\bullet}="7";(5,5)*{\bullet}="8";
(30,10)*{\bullet}="9";
(0,-2.5)*{}="11";(0,2.5)*{}="10";
{\ar@{-}@[blue]"1";"6"};{\ar@{-}@[blue]"1";"4"};{\ar@{-}@[red]"2";"4"};{\ar@{-}@[red]"2";"9"};
{\ar@{-}@[blue]"4";"7"};{\ar@{-}@[red]"4";"8"};{\ar@{..}"11";"10"};
\endxy
\end{array}$$
%$$\begin{array}{rc}
%(i) & \xy
%(-10,-12)*{\text{\scriptsize{$\left|\Psi^- \right\rangle_{12}$}}}="";
%(10,-12)*{\text{\scriptsize{$\left|\Psi^- \right\rangle_{34}$}}}="";
%(0,4)*{\text{\scriptsize{BS}}}="";
%(-10,-10)*{}="1";
%(10,-10)*{}="2";
%(-20,0)*{}="3";
%(0,-.5)*{}="4a";(0,.5)*{}="4b";
%(20,0)*{}="5";
%(-30,10)*{\bullet}="6";
%(-5,5)*{\bullet}="7";(5,5)*{\bullet}="8";
%(30,10)*{\bullet}="9";
%(0,-2.5)*{}="11";(0,2.5)*{}="10";
%{\ar@{-}"1";"6"};{\ar@{-}"1";"4a"};{\ar@{-}"2";"4a"};{\ar@{-}"2";"9"};
%{\ar@{-}"4b";"7"};{\ar@{-}"4b";"8"};{\ar@{..}"11";"10"};
%\endxy
%\\
%(ii) & 
%\xy
%(-10,-12)*{\text{\scriptsize{$\left|\Psi^- \right\rangle_{12}$}}}="";
%(10,-12)*{\text{\scriptsize{$\left|\Psi^- \right\rangle_{34}$}}}="";
%(0,4)*{\text{\scriptsize{M}}}="";
%(-10,-10)*{}="1";
%(10,-10)*{}="2";
%(-20,0)*{}="3";
%(-.5,0)*{}="4a";(.5,0)*{}="4b";
%(20,0)*{}="5";
%(-30,10)*{\bullet}="6";
%(-5,5)*{\bullet}="7";(5,5)*{\bullet}="8";
%(30,10)*{\bullet}="9";
%(0,-2.5)*{}="11";(0,2.5)*{}="10";
%{\ar@{-}"1";"6"};{\ar@{-}"1";"4a"};{\ar@{-}"2";"4b"};{\ar@{-}"2";"9"};
%{\ar@{-}"4a";"7"};{\ar@{-}"4b";"8"};{\ar@{..}"11";"10"};
%\endxy
%\end{array}$$
\caption{(i) The spacetime diagram for entanglement swapping of a 4-photon system.  (ii) The spacetime diagram for the 4-photon system with a mirror in place of the beam splitter.  In both cases there are two spacetime events, drawn in red and blue.}
\label{entanglement swapping}
\end{figure}

\subsection{(Delayed choice) entanglement swapping}

In entanglement swapping, two pairs of entangled photons, $1,2$ and $3,4$, are produced in the state $\left| \Psi^- \right\rangle$ \cite{HBGSSZ}.
Photons $2$ and $3$ then travel to a 50/50 beam splitter, interfere, and become entangled from a Bell state measurement.
This measurement causes particles $1$ and $4$ to become entangled, and disentangles the pairs $1,2$ and $3,4$.
The wavefunction of the 4-photon system is
\begin{equation} \label{es}
\begin{array}{l}
\left| \Psi^- \right\rangle_{12} \otimes \left| \Psi^- \right\rangle_{34} = \\
\ \ \ \ \ \ \ \ \ \ \frac 12 \left( \left| \Psi^+ \right\rangle_{14} \otimes \left| \Psi^+ \right\rangle_{23} - \left| \Psi^- \right\rangle_{14} \otimes \left| \Psi^- \right\rangle_{23} \right. \\
\ \ \ \ \ \ \ \ \ \ \left. - \left| \Phi^+ \right\rangle_{14} \otimes \left| \Phi^+ \right\rangle_{23} + \left| \Phi^- \right\rangle_{14} \otimes \left| \Phi^- \right\rangle_{23} \right),
\end{array}
\end{equation}
%\begin{widetext}
%\begin{equation} \label{es}
%\left| \Psi^- \right\rangle_{12} \otimes \left| \Psi^- \right\rangle_{34} =
%\frac 12 \left( \left| \Psi^+ \right\rangle_{14} \otimes \left| \Psi^+ \right\rangle_{23} - \left| \Psi^- \right\rangle_{14} \otimes \left| \Psi^- \right\rangle_{23} - \left| \Phi^+ \right\rangle_{14} \otimes \left| \Phi^+ \right\rangle_{23} + \left| \Phi^- \right\rangle_{14} \otimes \left| \Phi^- \right\rangle_{23} \right),
%\end{equation}
%\end{widetext}
where the state is represented in the initial physical basis on the left, and in the final physical basis on the right.
The spacetime diagram of the 4-photon system is given in Figure \ref{entanglement swapping}.ii.
Note that the support consists of two separate spacetime events; one looks like a $W$ and one looks like a $V$.

If a mirror is used instead of the beam splitter, then the pair $2,3$ undergoes a separable state measurement.
Consequently the entanglement of the original two pairs, $1,2$ and $3,4$, persists.
The spacetime diagram of this scenario is given in Figure \ref{entanglement swapping}.iii.
Note that the support again consists of two separate spacetime events, although of a different shape from entanglement swapping.

In (\ref{es}) there is no dependence on the spacetime location of the beam splitter.
Thus the choice of beam splitter or mirror could be made \textit{after} particles $1$ and $4$ have been measured.
This scenario, proposed by Peres \cite{P}, is known as \textit{delayed choice} entanglement swapping, and has been confirmed experimentally \cite{MZKUJBZ}.
The spacetime diagram of this scenario is given in Figure \ref{entanglement swapping}.iv, and is fundamentally the same as entanglement swapping without delayed choice in our spacetime framework.

\begin{Theorem} \label{main theorem}
Let $t$ be an inferred inertial time parameter.
Suppose two identical indivisible particles are unentangled for $t < 0$, and become maximally entangled at $t = 0$.
Then the union of their past ($t < 0$) inferred worldlines form a single spacetime event.
\end{Theorem}

Let $p$ be the inferred event at $t = 0$ where the two particles become entangled.
It is not clear if $p$ belongs to their common past spacetime event; see Proposition \ref{overlap} below.

\begin{proof}
Suppose two particles $a$ and $b$ in the state $\left| \psi \right\rangle \in \mathcal{H}$ interfere and become maximally entangled.
Choose a 2-dimensional subspace $\mathcal{H}'$ of $\mathcal{H}$, and an orthonormal basis $\left\{ \left| 0 \right\rangle, \left| 1 \right\rangle \right\}$ of $\mathcal{H}'$, such that $\left| \psi \right\rangle = \left| - \right\rangle$.
For $t < 0$, the two states form a separable bipartite state
$$\left| - \right\rangle \otimes \left| - \right\rangle = \left| -- \right\rangle = \frac 12 \left( \left| 00 \right\rangle + \left| 11 \right\rangle - \left| 01 \right\rangle - \left| 10 \right\rangle \right).$$

Recall the initial state in entanglement swapping (\ref{es}),
$$\left| \Psi^- \right\rangle_{12} \otimes \left| \Psi^- \right\rangle_{34} = \frac 12 \left( \left| 0101 \right\rangle + \left| 1010 \right\rangle - \left| 0110 \right\rangle - \left| 1001 \right\rangle \right).$$
If particles 1 and 4 are omitted, we obtain
$$\frac 12 \left( \left| 10 \right\rangle + \left| 01 \right\rangle - \left| 11 \right\rangle - \left| 00 \right\rangle \right).$$
This is precisely the state $\left| -- \right\rangle$, up to the global phase factor $-1$.
Assuming particles $1$ and $4$ cannot effect the interaction between their respective partners $2$ and $3$, we conclude that the spacetime support of particles $a$ and $b$ is the same as the spacetime support of particles $2$ and $3$ in Figure \ref{entanglement swapping}.i (with particles $1$ and $4$ omitted). 
\end{proof}

\section{Implications for the ontology of the wavefunction} \label{last section}

Throughout we assume that there is an ontic state space $\Lambda$.
The famous PBR theorem proves that the wavefunction is an ontic state (i.e., a real physical object) under the assumption that systems which are prepared independently have independent ontic states \cite{PBR}.
We will argue that in our spacetime framework, entanglement swapping invalidates this assumption.

We first establish notation.
Let $p(\lambda | P)$ be the probability distribution that an ontic state $\lambda \in \Lambda$ arises from a given preparation $P$, and let $p(k|M,\lambda)$ be the probability distribution of outcome $k$ given a measurement $M$ on $\lambda \in \Lambda$.

Suppose preparation $P$ produces the pure quantum state $\psi_P$, and the measurement $M$ of $\psi_P$ with outcome $k$ produces the quantum state $\phi_k$. 
Harrigan and Spekkens propose that the Born rule may be reformulated using the ontic state space $\Lambda$ \cite[Definition 1]{HS}:\footnote{For a general (possibly mixed) quantum state with density matrix $\rho$ and POVM element $E_k$ associated with outcome $k$ of $M$, the left side of (\ref{int}) equals $\operatorname{tr}\left(\rho E_k \right)$.}
\begin{equation} \label{int}
p(k| M,P) = \int_{\Lambda} d\lambda \, p(k| M, \lambda) \, p( \lambda | P) = \left\langle \phi_k | \psi_P \right\rangle^2.
\end{equation}

Furthermore, they define an ontological model to be $\psi$-ontic (i.e., wavefunctions are real) if and only if for each pair of distinct quantum states $\psi_1$ and $\psi_2$ with preparations $P_1$ and $P_2$, we have \cite[Definitions 4 and 5]{HS}
\begin{equation} \label{lambda in}
p(\lambda | P_1) \, p(\lambda | P_2) = 0, \ \ \ \ \forall \lambda \in \Lambda.
\end{equation}
Otherwise the model is $\psi$-epistemic (i.e., wavefunctions merely specify our statistical knowledge of real underlying states).

Recall the argument of the PBR theorem:
Assume to the contrary that quantum theory is equivalent to a $\psi$-epistemic model with an ontic state space $\Lambda$.
Then there are distinct quantum states for which (\ref{lambda in}) does not hold (on a set $\Lambda_0 \subset \Lambda$ of nonzero measure).
For simplicity, we may take these states to be $\left| 0 \right\rangle$ and $\left| - \right\rangle$, with respective preparations $P_0$ and $P_-$.
(Here we take the second state to be $\left| - \right\rangle$ rather than $\left| + \right\rangle$, as used in \cite{PBR}, so that our analysis may be related to entanglement swapping.)

Now suppose Alice and Bob each independently prepare a particle in one of these two states.
Denote by $P_1$ the preparation of Alice, and by $P_2$ the preparation of Bob.
The initial bipartite state $\left| \psi_P \right\rangle$ with preparation $P = (P_1,P_2)$ is then one of the four separable states
\begin{equation} \label{initial state}
\left|00 \right\rangle, \ \ \left|0- \right\rangle, \ \ \left| -0 \right\rangle, \ \ \left| -- \right\rangle.
\end{equation}
An entanglement measurement $M$ is then preformed, projecting $\left| \psi_P \right\rangle$ onto one of the four states 
$$\begin{array}{lcl}
\left| \phi_1 \right\rangle & = & \frac{1}{\sqrt{2}}\left( \left| 01 \right\rangle + \left| 10 \right\rangle \right)\\
\left| \phi_2 \right\rangle & = & \frac{1}{\sqrt{2}}\left( \left| 0+ \right\rangle + \left| 1- \right\rangle \right)\\
\left| \phi_3 \right\rangle & = & \frac{1}{\sqrt{2}}\left( \left| -1 \right\rangle + \left| +0 \right\rangle \right)\\
\left| \phi_4 \right\rangle & = & \frac{1}{\sqrt{2}}\left( \left| -+ \right\rangle + \left| +- \right\rangle \right).
\end{array}$$

The preparation independence assumption implies that
\begin{equation} \label{prep ind}
p(\lambda_1, \lambda_2 | P_1,P_2) = p(\lambda_1 | P_1) \, p(\lambda_2 | P_2).
\end{equation}
Thus for each $1 \leq k \leq 4$,
$$\begin{array}{rcl}
\left\langle \phi_k | \psi_P \right\rangle^2 & \stackrel{(\textsc{i})}{=} & \int_{\Lambda \times \Lambda} d\lambda \, p(k| M,\lambda) \, p(\lambda | P) \\
& \stackrel{(\textsc{ii})}{=} & \int_{\Lambda} \int_{\Lambda} d\lambda_1 d\lambda_2 \, p(k| M, \lambda_1, \lambda_2) \, p(\lambda_1, \lambda_2 | P_1,P_2)\\
& \stackrel{(\textsc{iii})}{=} & \int_{\Lambda} \int_{\Lambda} d\lambda_1 d\lambda_2 \, p(k| M,\lambda_1, \lambda_2) \, p(\lambda_1 | P_1) \, p(\lambda_2 | P_2)\\
& \stackrel{(\textsc{iv})}{\not =} & 0,
\end{array}$$
where (\textsc{i}) holds by (\ref{int}); (\textsc{ii}) holds since $\lambda = (\lambda_1,\lambda_2) \in \Lambda \times \Lambda$; (\textsc{iii}) holds by (\ref{prep ind}); and (\textsc{iv}) holds by (\ref{lambda in}) and the assumption that $\psi_1$ and $\psi_2$ are $\psi$-epistemic. 
However, for each choice of preparation $P = (P_1, P_2)$ there is an outcome $k$ for which $\left\langle \phi_k | \psi_P \right\rangle = 0$, namely
$$\left\langle \phi_1 | 00 \right\rangle = \left\langle \phi_2 | 0- \right\rangle = \left\langle \phi_3 | -0 \right\rangle = \left\langle \phi_4 | -- \right\rangle = 0.$$
We thus arrive at a contradiction.

\begin{Proposition}
Suppose the strong postulate of time holds.
Then the preparation assumption in the PBR theorem does not hold by entanglement swapping.
\end{Proposition}

\begin{proof} 
Suppose the initial state is $\left| 00 \right\rangle$ or $\left| -- \right\rangle$.
By Theorem \ref{main theorem}, these states are each supported on two spacetime events as shown in Figure \ref{PBR}.
Thus
$$p(\lambda_1, \lambda_2 | P_0,P_0) \not = p(\lambda_1 | P_0) \, p(\lambda_2 | P_0)$$
and
$$p(\lambda_1, \lambda_2 | P_-,P_-) \not = p(\lambda_1 | P_-) \, p(\lambda_2 | P_-).$$
Therefore the preparation independence assumption (\ref{prep ind}) does not hold.
\end{proof}

\begin{figure}
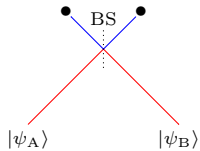

$$\xy
(-10,-12)*{\text{\scriptsize{${\left| \psi_{\operatorname{A}} \right\rangle}$}}}="";
(10,-12)*{\text{\scriptsize{${\left| \psi_{\operatorname{B}} \right\rangle}$}}}="";
(0,4)*{\text{\scriptsize{BS}}}="";
(-10,-10)*{}="1";
(10,-10)*{}="2";
(-20,0)*{}="3";
(0,0)*{}="4";
(20,0)*{}="5";
(-5,5)*{\bullet}="7";(5,5)*{\bullet}="8";
(0,-2.5)*{}="11";(0,2.5)*{}="10";
{\ar@{-}@[red]"1";"4"};{\ar@{-}@[red]"2";"4"};
{\ar@{-}@[blue]"4";"7"};{\ar@{-}@[blue]"4";"8"};{\ar@{..}"11";"10"};
\endxy$$
%$$\xy
%(-10,-12)*{\text{\scriptsize{${\left| \psi_{\operatorname{A}} \right\rangle}$}}}="";
%(10,-12)*{\text{\scriptsize{${\left| \psi_{\operatorname{B}} \right\rangle}$}}}="";
%(0,4)*{\text{\scriptsize{BS}}}="";
%(-10,-10)*{}="1";
%(10,-10)*{}="2";
%(-20,0)*{}="3";
%(0,-.5)*{}="4a";(0,.5)*{}="4b";
%(20,0)*{}="5";
%(-5,5)*{\bullet}="7";(5,5)*{\bullet}="8";
%(0,-2.5)*{}="11";(0,2.5)*{}="10";
%{\ar@{-}"1";"4a"};{\ar@{-}"2";"4a"};
%{\ar@{-}"4b";"7"};{\ar@{-}"4b";"8"};{\ar@{..}"11";"10"};
%\endxy$$
\caption{The spacetime diagram for the 2-photon system with inital state $\left| 00 \right\rangle$ or $\left| -- \right\rangle$ in the PBR theorem.  The independence assumption in the theorem does not hold by the spacetime diagram for entanglement swapping.  Note that there are two spacetime events, drawn in red and blue.}
\label{PBR}
\end{figure}

In the following proposition we show that the wavefunction is epistemic if spacetime events are topologically closed; or equivalently, if the wavefunction is ontic, then spacetime events are not topologically closed.
A primary difficulty with an epistemic interpretation of the wavefunction is that interference effects require an alternative explanation.

\begin{Proposition} \label{overlap}
If spacetime events are closed sets, then an indivisible state is collapsed along its entire support in inferred spacetime.
In particular, its wavefunction is epistemic.
\end{Proposition}

\begin{proof}
Suppose spacetime events are closed sets.
Fix an inertial frame with inferred time parameter $t \in \mathbb{R}$.
Consider an indivisible state $\psi(t) \in \mathcal{H}$ that is prepared at $t = 0$, evolves by a local unitary transformation $U(t) \in \operatorname{GL}(\mathcal{H})$, 
$$\psi(t) = U(t) \psi(0),$$
and is measured at $t = 1$,
$$\psi(1) = \sum_i \psi_i(1) \mapsto \psi_j(1).$$
Let $X$ be inferred spacetime, and denote by $p \in X$ the support of $\psi(t)$.
Since $p$ is an event and events are closed, the inferred time interval of $p$ is $\left[0, 1 \right]$.

In standard quantum mechanics, $\psi$ collapses at $t = 1$.
However, we may assume $\psi$ collapses at any inferred time $t_c$ in the interval $\left[0, 1 \right]$,
$$\psi(t_c) \mapsto \psi_j(t_c).$$
But then $\psi$ is collapsed at the event $p$ since measurement is a non-invertible transformation.
In particular, $\psi$ is collapsed along its entire support in $X$.
(Note that this conclusion does not hold if $\psi$ is assumed to collapse at $t = 1$ and the inferred time interval of $p$ is either the open interval $\left(0,1 \right)$ or half-open interval $\left[0,1 \right)$.)
\end{proof}

\begin{Remark} \rm{
Proposition \ref{overlap} implies that if spacetime events are closed, then quantum superposition is merely epistemic, and is a natural consequence of the fact that events may not be zero-dimensional. 
One may be tempted to say that, in this framework, $\psi$ collapses at the moment $\psi$ is prepared, rather than at the moment $\psi$ is measured. 
However, this view is misleading.
Indeed, the `moment of preparation' \textit{is} the `moment of measurement' for $\psi$.
These moments are not distinct in $\psi$'s frame, but only in the frame of the experimenter.
This may be partly why quantum mechanics appears so strange: the experimenter and the state she is measuring disagree on what a moment of time is. 
}\end{Remark}

\section{Conclusion}

The theory we propose suggests that many of the strange properties of quantum theory arise from the geometry of spacetime itself, namely, that there are positive dimensional events.
This, in turn, is a simple consequence of the identity of indiscernibles applied to time.
The main directions for future research are: (1) understanding the spacetime support of non-maximally entangled states with well-defined positions; (2) understanding the spacetime support of indivisible states with uncertain position (here we expect interference and quantum field theory to enter); (3) a mathematical theory of differential geometry where manifold-like objects may have positive dimensional points; and (4) cosmological consequences of the (strong) postulate of time.

\begin{acknowledgements}
 We thank Vinesh Solanki and Sam Pallister for useful discussions.
\end{acknowledgements}

%\bibliographystyle{hep}
%\def\cprime{$'$} \def\cprime{$'$}
%\bibliography{mybib}

\end{document}